\documentclass[aps, pra, twocolumn, superscriptaddress, 10pt]{revtex4-2}
\usepackage{amsthm}
\usepackage{amsmath,bm}
\usepackage{amssymb}
\usepackage{amsfonts}
\usepackage{graphicx}
\usepackage{txfonts}
\usepackage{xcolor}
\usepackage{braket}
\usepackage{bbm}
\usepackage{subcaption}
\usepackage{changes}
\usepackage{ragged2e}
\usepackage[colorlinks=true,linkcolor=blue,citecolor=blue,urlcolor=blue]{hyperref}
\usepackage[capitalise]{cleveref}
\usepackage{booktabs}
\usepackage{diagbox}
\usepackage{makecell}
\usepackage{todonotes}

\newcommand{\ignore}[1]{} 

\newcounter{SaveEqnCntr}

\newcommand{\be}{\begin{equation}}
\newcommand{\ee}{\end{equation}}
\newcommand{\ba}{\begin{eqnarray}}
\newcommand{\ea}{\end{eqnarray}}

\def\>{\rangle}
\def\<{\langle}

\newcommand{\Tr}{\text{Tr}}

\newtheorem{theorem}{Theorem}

\newtheorem{definition}{Definition}

\newtheorem{lemma}{Lemma}

\begin{document}
	
\title{Detection of quantum imaginarity using moments and its interferometric realization}

\author{Sudip Chakrabarty}
\email{sudip27042000@gmail.com}
\affiliation{S. N. Bose National Centre for Basic Sciences, Block JD, Sector III, Salt Lake, Kolkata 700 106, India}

\author{Saheli Mukherjee}
\email{mukherjeesaheli95@gmail.com}
\affiliation{S. N. Bose National Centre for Basic Sciences, Block JD, Sector III, Salt Lake, Kolkata 700 106, India}

\author{Ananda G. Maity}
\email{anandamaity289@gmail.com}
\affiliation{School of Physical Sciences, Indian Institute of Technology Goa, Ponda 403401, Goa, India}

\author{Bivas Mallick}
\email{bivasqic@gmail.com}
\affiliation{S. N. Bose National Centre for Basic Sciences, Block JD, Sector III, Salt Lake, Kolkata 700 106, India}

\begin{abstract}
Complex numbers, intrinsic to the formulation of quantum theory, play a pivotal role in enabling advantages across a broad range of quantum information-processing tasks. Despite their fundamental importance, practical and scalable criteria for detecting quantum imaginarity remain relatively underexplored, particularly methods that enable its identification with reduced experimental overhead. In this work, we propose a realistic and experimentally feasible method to detect quantum imaginarity using moment-based approach. Our framework relies on experimentally accessible moments of the Kirkwood–Dirac quasiprobability distribution, enabling scalable detection in high-dimensional systems without requiring full state tomography. We then present illustrative examples to support our detection scheme. Finally, we present an interferometric scheme for measuring these moments, paving the way for experimental implementation of our detection protocol.
\end{abstract}
\maketitle
\section{Introduction}
Real numbers form the backbone of classical physics, offering a complete and consistent framework for the description of physical phenomena 
\cite{stueckelberg1960quantum,araki1980characterization,wootters2012entanglement,hardy2012limited,PhysRevA.87.052106,PhysRevLett.102.020505,varun2023real}. In contrast, the status of complex numbers in physical theories, particularly in quantum mechanics, has long been a subject of conceptual debate \cite{PhysRevLett.102.020505, moretti2017quantum,renou2021quantum,PhysRevLett.128.040403}. Although complex numbers are routinely employed in the mathematical formalism of quantum theory, their necessity is often assumed rather than critically examined. For example, the authors in \cite{PhysRevLett.102.020505} show that the outcomes of any multipartite experiment can be reproduced by a real quantum theory when the parties share arbitrary real quantum states. This was followed by proving the distinctness of real and complex quantum theory in network scenarios with independent sources and measurements \cite{renou2021quantum}. Recent studies show that by moving beyond the tensor product structure of Hilbert space for composite systems, one can have an equivalent formalism of the theory in terms of real numbers \cite{4k13-sdjh}. This long-standing debate has motivated sustained interest in determining whether complex numbers are intrinsically required for an accurate description of quantum phenomena \cite{sarkar2026gap,3fv7-p8cs, PhysRevA.111.062401, PhysRevLett.128.040403, PhysRevLett.129.140401, PhysRevLett.128.040402}, or whether quantum mechanics can be consistently reformulated within a purely real number framework \cite{myrheim1999quantum,koh2018quantum,PhysRevA.87.052106,moradi2026new}.

To gain deeper insight into the significance of complex numbers in quantum mechanics, Hickey and Gour introduced the resource theory of imaginarity in 2018 \cite{hickey2018quantifying}, which was subsequently further developed and extended by Wu et al. in 2021 \cite{wu2021operational}. Within this framework, quantum states that possess complex-valued matrix elements with respect to a fixed reference basis are regarded as resourceful, as they can be exploited to yield advantages in a wide range of quantum information processing tasks, including quantum state and channel discrimination tasks \cite{wu2021operational,wu2024resource}, quantum communication \cite{PhysRevA.111.062401}, multiparameter metrology \cite{miyazaki2022imaginarity}, non-local advantages \cite{wei2024nonlocal,datta2025efficient} and many others \cite{kedem2012using,zhu2021hiding,li2022brukner,sajjan2023imaginary,zhang2024broadcasting,zhang2024can,roden_2016,biswas2026increasing,bera2026strong}.

Despite the many promising applications of quantum imaginarity, a fundamental and practically relevant question is to reliably determine the presence of imaginarity in a quantum state, so that it can be effectively harnessed as a resource for achieving quantum advantage over classical analogues. Additionally, imaginarity is independent of other well-studied quantum resources such as coherence \cite{winter2016operational,streltsov2017colloquium,streltsov2017structure} and entanglement \cite{horodecki2009quantum}. A state can exhibit coherence while remaining entirely real in a given reference basis, and entanglement does not, in general, necessitate the presence of intrinsically complex phases. Consequently, witnesses developed for coherence \cite{napoli2016robustness,ren2017quantitative,ma2021detecting} or entanglement \cite{Guhne2009,chruscinski2014entanglement} are insufficient to certify imaginarity. Dedicated detection methods are required to identify this specifically complex feature. Although significant research has been carried out in this direction from different perspectives~\cite{xue2021quantification,wu2021resource,xu2024quantifying,fernandes2024unitary,zhang2025imaginarity,guo2025quantifying, liang2026imaginaritywitness}, developing scalable and experimentally feasible detection methods is particularly important for high-dimensional systems, where full state tomography~\cite{haah2017sample, odonnell2016efficient, d2003quantum_tomography} is impractical. In this context, we propose a detection criterion that enables the practical identification of imaginarity while significantly reducing experimental overhead. 

Here, we develop a strategy for detecting quantum imaginarity based on the moments of the Kirkwood–Dirac (KD) distribution~\cite{kirkwood1933quantum,dirac1945analogy,arvidsson2024properties,PhysRevLett.127.190404,PhysRevA.109.012215,fan2024resource}. In this approach, we first apply a Y-twirling map to remove the real coherence while retaining contributions from imaginary components, thereby isolating genuinely complex features of the state. We then examine the positivity of the resulting KD distribution: real states satisfy specific positivity constraints, and any violation implies imaginarity. However, to avoid full reconstruction of the KD distribution and thereby to ensure experimental scalability, we formulate the detection criterion directly in terms of the moments of the KD distribution \cite{lkcb-nqqb}. We then present illustrative examples to demonstrate the effectiveness of our detection scheme. Next, to demonstrate the practical relevance of our moment-based approach, we develop a methodology that enables detection of quantum imaginarity through the visibility obtained in a Mach–Zehnder interferometer. While our approach builds upon several 
well-established ingredients including KD quasiprobability distribution, positivity properties of moment matrices, and interferometric visibility estimation, the main contribution of the present work is the development of a unified operational framework that combines these existing independent ingredients to address the problem of detecting quantum imaginarity.

The paper is organized as follows. In section \ref{s2}, we briefly discuss the essential mathematical preliminaries, as well as the moment criteria proposed in earlier works for entanglement detection. In section \ref{s3}, we present our framework for the detection of quantum imaginarity along with an explicit example. We then discuss a scheme for the realization of these moments and a detection method of imaginarity from the visibility obtained in an interference pattern in section \ref{s4}. Finally, in section \ref{s5}, we summarize our main findings along with some future perspectives.

\section{Preliminaries}\label{s2}

In this section, we introduce the basic concepts, definitions, and notations that are necessary for the subsequent discussions. Unless stated otherwise, throughout the paper, we adopt the standard notation and terminology widely used in quantum information theory.
\subsection{Quantum imaginarity }\label{s2A}
Let $\{ \ket{a_i} \}_{i=0}^{d-1}$ denote a fixed orthonormal reference basis of a $d$-dimensional Hilbert space. 
A density operator 
\(
\rho = \sum_{i,j} \rho_{ij}\, \ket{a_i}\bra{a_j}
\)
is said to be \emph{real} in this basis if and only if all its matrix elements are real, i.e.,
\begin{equation}
\rho_{ij} \in \mathbb{R} \qquad \forall\, i,j.
\end{equation}
In other words, $\rho$ is real if and only if it is invariant under transposition in the reference basis, i.e. $\rho = \rho^{T}.$

Alternatively, a quantum state $\rho$ is said to possess \emph{imaginarity} whenever at least one off-diagonal matrix element has a nonzero imaginary component,
\begin{equation}
\exists\, i \neq j \quad \text{such that} \quad \mathrm{Im}(\rho_{ij}) \neq 0.
\end{equation}
In this sense, imaginarity captures the intrinsically complex structure of quantum coherence~\cite{streltsov2017colloquium} in the chosen reference basis. Notably, the real and imaginary parts of coherence are not operationally equivalent: they can affect physical processes in qualitatively different ways~\cite{wu2021operational,wu2024resource,PhysRevA.111.062401,miyazaki2022imaginarity,wei2024nonlocal,datta2025efficient,kedem2012using,zhu2021hiding,li2022brukner,sajjan2023imaginary,zhang2024broadcasting,zhang2024can,roden_2016}, underscoring the importance of treating imaginarity as an independent feature of quantum states.

In operator-space language, any density matrix can be decomposed as
$\rho = \rho_{\mathrm{real}} + \rho_{\mathrm{imag}},$ where $\rho_{\mathrm{real}} = {(\rho + \rho^{T})}/{2}$ and 
$\rho_{\mathrm{imag}} = {(\rho - \rho^{T})}/{2}$. The real part ($\rho_{\mathrm{real}}$) is symmetric, while the imaginary part ($\rho_{\mathrm{imag}}$) is antisymmetric in the reference basis. In higher dimensions, the imaginary sector is spanned by the antisymmetric generators
\begin{equation}
Y_{pq} = i\left( \ket{a_p}\bra{a_q} - \ket{a_q}\bra{a_p} \right),
\qquad p<q, \label{eq:antisymmetric_generators}
\end{equation}
which form a subset of the generators of $\mathfrak{su}(d)$.

In the resource theory of quantum imaginarity~\cite{hickey2018quantifying, wu2021operational}, the collection of the \emph{real} states forms the set of \emph{free states}. It is instructive to define the \emph{free operations} corresponding to physical transformations of quantum systems, under which imaginarity cannot be generated or increased. Usually, a quantum operation $\Lambda$ is described by Kraus operators $\{K_\mu\}$ satisfying $\sum_\mu K_\mu^\dagger K_\mu = \mathbb{I}$, with
\(
\Lambda[\rho] = \sum_\mu K_\mu \rho K_\mu^\dagger.
\)
Considering imaginarity as the resource, the \emph{free operations} are identified as those admitting a Kraus representation whose matrix elements are \emph{real} in the reference basis $\{ \ket{a_i}\}$, i.e.,
\begin{equation}
\bra{a_i} K_\mu \ket{a_j} \in \mathbb{R}
\qquad \forall\, i,j,\mu.
\end{equation}
A natural quantitative measure of imaginarity is given by the $\ell_1$-norm of the imaginary part of the density matrix~\cite{Chen2023}:
\begin{equation}
\mathcal{M}_{\ell_1}(\rho)
=
\sum_{i\neq j} \, \big| \mathrm{Im}(\rho_{ij}) \big|.
\label{eq:l1_imaginarity}
\end{equation}
This measure satisfies:
\begin{enumerate}
\item $\mathcal{M}_{\ell_1}(\rho)=0$, if $\rho $ is a \emph{free} (\emph{real}) state.
\item $\mathcal{M}_{\ell_1} (\Lambda_F(\rho)) \leq \mathcal{M}_{\ell_1} (\rho)$, if $\Lambda_F$ is a \emph{free (real) operation}.
\end{enumerate}

\subsection{Kirkwood-Dirac quasiprobability distribution}\label{s2B}

We consider a Hilbert space of dimension $d$, on which all relevant operators are defined. 
Let $\{\ket{a_i}\}$ and $\{\ket{b_k}\}$ be two orthonormal basis sets associated with the spectral resolutions of the observables
\[
A=\sum_i a_i\,\ket{a_i}\bra{a_i}, 
\qquad 
B=\sum_k b_k\,\ket{b_k}\bra{b_k}.
\]

Given a density operator $\rho$, the Kirkwood-Dirac (KD) quasiprobability distribution corresponding to these two bases is defined as:
\begin{equation}
Q_{ik}(\rho)
\equiv \langle b_k|a_i\rangle\,\bra{a_i}\rho\ket{b_k}
= \mathrm{Tr}\!\left(\Pi^b_k \Pi^a_i \rho\right),
\label{Eq:KD}
\end{equation}
with $\Pi^a_i=\ket{a_i}\bra{a_i}$ and $\Pi^b_k=\ket{b_k}\bra{b_k}$. The quantities $Q_{ik}(\rho)$ need not be real or positive and therefore, in general, cannot be interpreted as an ordinary probability distribution.
Despite this, they satisfy basic consistency relations:
\begin{align}
\sum_{i,k} Q_{ik} = 1, \;\;
\sum_k Q_{ik} = \mathrm{Tr}(\Pi^a_i \rho), \;\;
\sum_i Q_{ik} = \mathrm{Tr}(\Pi^b_k \rho),
\end{align}
ensuring that the marginals reproduce the Born rule probabilities for measurements in the $\{\ket{a_i}\}$ and $\{\ket{b_k}\}$ bases \cite{Kolmogorov51}. However, when the two bases coincide (i.e., $\ket{a_i}=\ket{b_k}$), the above expression reduces to
\(
Q_{ik}(\rho)=\mathrm{Tr}(\Pi^a_i \rho)\,\delta_{ik},
\)
which is a probability distribution.

The above KD framework also extends straightforwardly to a sequence of observables \cite{Yunger18,gonzalez2019out,Razieh19,ArvidssonShukur2020}. 
For observables
\[
A^{(r)}=\sum_i a^{(r)}_i\,\Pi^{a^{(r)}}_i,
\qquad r=1,\ldots,l,
\]
one defines the extended KD distribution as
\begin{equation}
Q^\star_{i_1,\ldots,i_l}(\rho)
=
\mathrm{Tr}\!\left(
\Pi^{a^{(l)}}_{i_l}\cdots
\Pi^{a^{(1)}}_{i_1}\,\rho
\right).
\label{KDExt}
\end{equation}
One may note that the usual two-observable case can be recovered as a special case by setting $l=2$. 
The relevance of the extended KD distribution to our analysis will be discussed in Sec.~\ref{s3}.

An important aspect of the extended KD distribution is that it encodes complete information about the quantum state. 
In fact, provided the relevant overlaps are nonzero, the density operator can be reconstructed from the extended KD coefficients as
\begin{equation}
\rho
=
\sum_{i_1,\ldots,i_l}
\frac{
\ket{a^{(1)}_{i_1}}\bra{a^{(l)}_{i_l}}
}{
\langle a^{(l)}_{i_l}|a^{(1)}_{i_1}\rangle
}
\,Q^\star_{i_1,\ldots,i_l}(\rho).
\label{KDDecom}
\end{equation}

KD quasiprobabilities have found several applications in quantum information theory, quantum foundations, and thermodynamics, where they are used to analyze nonclassical features and operational constraints 
\cite{Dressel15,Yunger18,gonzalez2019out,Razieh19,ArvidssonShukur2020,kunjwal2019anomalous,pusey2014anomalous,dressel2011experimental,lostaglio2020certifying,Upadhyaya24}. 
For a detailed review, see Ref.~\cite{arvidsson2024properties}.

\subsection{Partial transposition moments} \label{s2C}
This subsection provides an introduction to the concept of \emph{moments} which we employ as a tool for detecting imaginarity. In particular, our approach is inspired by the moments of the partially transposed density matrix, commonly referred to as \emph{PT moments}~\cite{PhysRevLett.109.130502}. The PT moments are defined as
\begin{equation}
p_n := \mathrm{Tr}\!\left[(\rho^{T_B}_{AB})^n\right], \qquad n \geq 1,
\label{PTmoments}
\end{equation}
 
By construction, $p_1 = 1$, while $p_2$ corresponds to the purity of $\rho^{T_B}_{AB}$. In \cite{PhysRevLett.125.200501}, the authors proposed an entanglement detection criterion that relies solely on the first three PT moments. In particular, they proved that all PPT states satisfy
\begin{equation}
p_2^2 \leq p_3.
\end{equation}
Violation of this inequality implies that the underlying state is entangled, which is known as the $p_3$-PPT criterion. 

Further, higher-order moments have also been utilized to formulate criteria that capture finer spectral features and thereby enable more sensitive detection of entanglement \cite{Neven2021,yu2021optimal}. In particular, higher-order PT moments ($n \geq 4$) can be systematically exploited using Hankel matrices constructed from the moment sequence \cite{yu2021optimal}. Given the vector $\mathbf{p} = (p_1,p_2,\ldots,p_{2m+1})$, the corresponding $(m+1)\times(m+1)$ Hankel matrix is defined as
\begin{equation}
[H_m(\mathbf{p})]_{jk} := p_{j+k+1}, \quad j,k = 0,\ldots,m.
\label{Hankelmatrices}
\end{equation}
A necessary condition for separability is then given by
\begin{equation}
\det[H_m(\mathbf{p})] \geq 0, \quad \forall \;m \in \mathbb{N}
\label{Hankelmatrixcondition}
\end{equation}

An important practical advantage of PT moments is that they can be estimated using classical shadows~\cite{aaronson2018shadow,huang2020predicting,cieslinski2024analysing,PhysRevA.109.022247,PhysRevA.111.032406,mukherjee2025detecting,mallick2025higher,mallick2026detection}.
Moreover, since these criteria are state-independent, they are applicable even when little prior information about the system is available.

\subsection{Overview of a Mach–Zehnder interferometer}
Consider a Mach-Zehnder interferometer where the arms of the interferometer represent the path degrees of freedom. This is spanned by a two-dimensional Hilbert space with the corresponding basis denoted by $\{\ket{0}, \ket{1}\}$. The path elements of the interferometer are represented in this basis, whereas the quantum state whose imaginarity is to be detected constitutes the internal degrees of freedom. The composite system is thus represented by a tensor product of the path and internal degrees of freedom. The typical elements of an interferometer include two $50$-$50$ beam splitters, a phase shifter, a controlled unitary, two mirrors, and a measurement in the computational basis. The unitary gates representing these elements are given below.
\begin{itemize}
\item $50$-$50$ beam splitters $\Leftrightarrow$ Hadamard gate $H= \frac{1}{\sqrt{2}} \begin{pmatrix}
1 & 1 \\
1 & -1
\end{pmatrix}$. The corresponding unitary is $H \otimes \mathbb{I}$.
\item A phase shifter that shifts the phase of the path $\ket{0}$ by an angle $\theta$. The corresponding unitary is represented by $(e^{i \theta} \ket{0}\bra{0} \otimes \mathbb{I} + \ket{1}\bra{1} \otimes \mathbb{I})$.
\item A controlled unitary $U$ with the path (internal) degrees of freedom acting as the control (target). ($\ket{0}\bra{0} \otimes \mathbb{I} + \ket{1}\bra{1} \otimes U$) is the unitary.
\item Mirrors are represented by $\sigma_x = \begin{pmatrix}
0 & 1 \\
1 & 0
\end{pmatrix}$.
\end{itemize}
The overall evolution in an interferometer is given by the unitary \cite{PhysRevLett.85.2845, Kanjilal_2023}
\begin{equation}
\begin{split}
    U_{total} =& (H \otimes \mathbb{I})(\sigma_{x} \otimes \mathbb{I})(\ket{0}\bra{0} \otimes \mathbb{I} + \ket{1}\bra{1} \otimes U) \\ & (e^{i \theta} \ket{0}\bra{0} \otimes \mathbb{I} + \ket{1}\bra{1} \otimes \mathbb{I})(H \otimes \mathbb{I}).
    \end{split}
\end{equation}
If the initial state is $\rho_i=\ket{0}\bra{0}\otimes \rho$, the final state after evolving through the interferometer is $\rho_f=U \rho_i U^{\dagger}$. After a detailed calculation, we get
\begin{equation} \label{finalstate}
\begin{split}
    \rho_f=&\frac{1}{2}(\ket{+}\bra{+}\otimes U \rho U^{\dagger}+e^{-i\theta}\ket{+}\bra{-}\otimes U\rho \\ & +e^{i\theta}\ket{-}\bra{+}\otimes \rho U^{\dagger}+\ket{-}\bra{-}\otimes\rho).
\end{split}
\end{equation}
The intensity of the interference along the path $\ket{0}$ is 
\begin{align}
    I  &=\Tr[(\ket{0}\bra{0}\otimes \mathbb{I})\rho_f] \nonumber\\
        &=\frac{1}{2}(1+Re[\Tr(U \rho) e^{-i \theta}]), \label{intensity}
\end{align}
where $\Tr(U \rho)$ is in general, a complex quantity. Expressing $\Tr(U \rho)$ as $|\Tr(U \rho)|e^{i \chi}$, where $\chi=arg[\Tr(U \rho)]$, Eq.~\eqref{intensity} simplifies to
    \begin{equation} \label{intensity1}
        I=\frac{1}{2}[1+|\Tr(U \rho)| \cos{(\chi-\theta})]
    \end{equation}
    The visibility of an interference pattern is given by 
    \begin{equation} \label{visibility}
        V=\frac{\underset{\theta}{\max}\, I-\underset{\theta}{\min}\, I}
{\underset{\theta}{\max}\, I+\underset{\theta}{\min}\, I}
  \end{equation}
   After calculating the maximum and minimum intensities from Eq.~\eqref{intensity}, we get  
  \begin{equation} \label{visibility1}
      V=|\Tr(U \rho)|.
  \end{equation}
  Eq.~\eqref{visibility1} suggests that by choosing a suitable unitary, the desired visibility can be obtained. Any nonzero visibility reflects the persistence of wave-like behavior and therefore signals nonclassicality. Figure~\ref{fig:inteferometer} represents a schematic diagram of the setup.

\begin{figure}[t]
\centering

\includegraphics[width=0.9\linewidth]{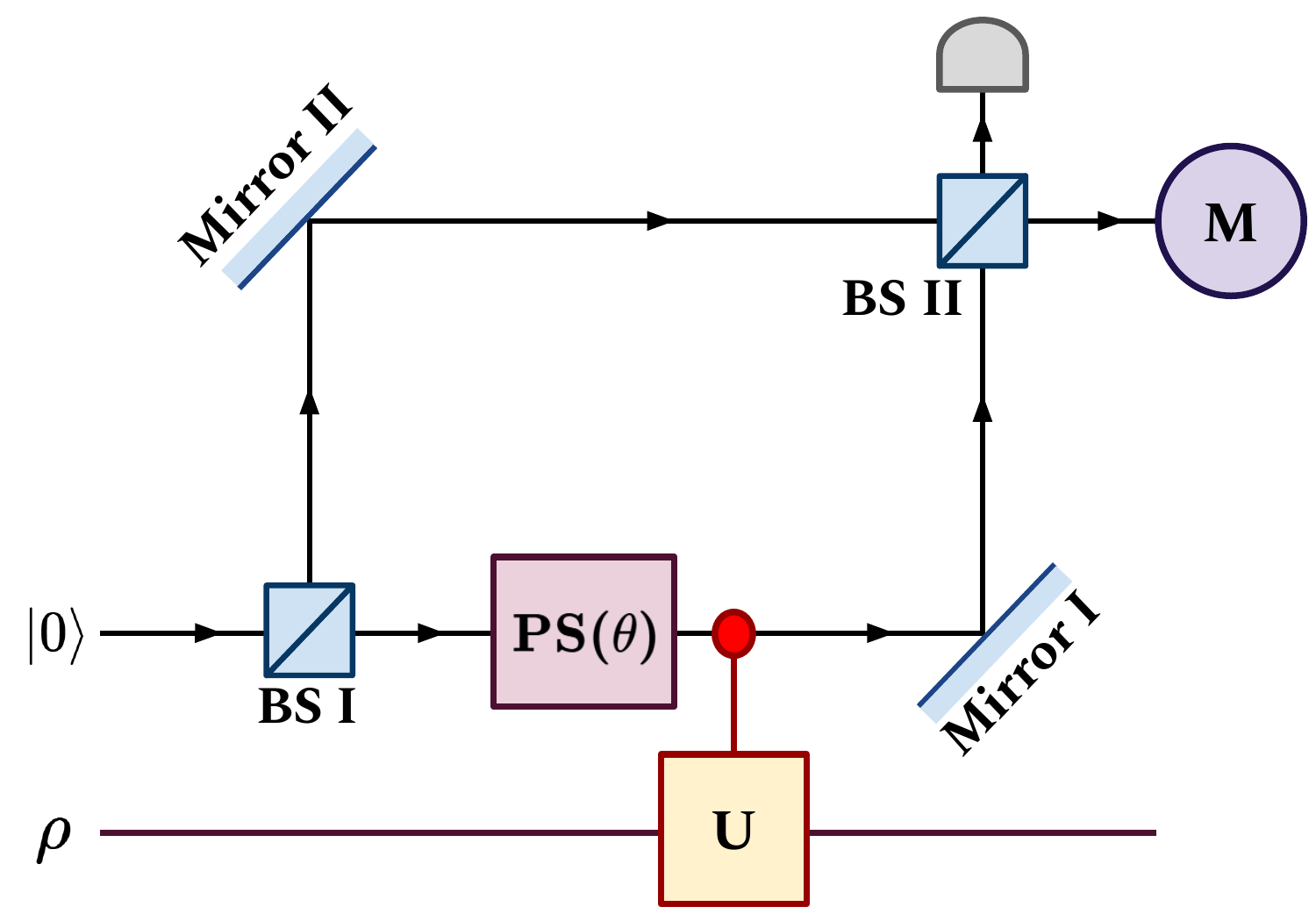}
\caption{ 
\justifying 
\small{
Schematic representation of the Mach–Zehnder interferometric setup used to estimate the visibility $V$. The path degree of freedom of the interferometer acts as the control qubit, while the internal system carrying the state $\rho$ serves as the target. The setup consists of two $50$-$50$ beam splitters (BS I and BS II), two mirrors, a phase shifter (PS) introducing a phase $\theta$ in the arm along $\ket{0}$, and a controlled unitary (U). A measurement (M) in the computational basis at the output port along the path $\ket{0}$ gives the visibility $V$ along path $\ket{0}$.
}}
\label{fig:inteferometer}
\end{figure}

\section{Detection of imaginarity using moments} \label{s3}
In this section, we establish a framework for detecting quantum imaginarity using moment-based approach. To this end, we first review the notion of moments of the extended KD distribution, which was introduced recently in Ref.~\cite{lkcb-nqqb}. Here we consider two orthonormal bases $\{\ket{a_i}\}$ and $\{\ket{b_k}\}$, where $\{\ket{b_k}\}$ is mutually unbiased basis with respect to $\{\ket{a_i}\}$. In this case, the elements of the extended KD distribution for the quantum state $\rho$ take the form:
\begin{equation}
   Q^{\star}_{i,j,k} (\rho) = \langle a_j | b_k \rangle \langle b_k | a_i \rangle \langle a_i | \rho | a_j \rangle \label{eq:extended_KD_definition}.
\end{equation} 

Now to detect quantum imaginarity, we proceed in three steps. First, we introduce a symmetry-based twirling operation, referred to here as the $Y$-twirl, which removes all real coherence of a quantum state while preserving its imaginary components. Next, we exploit the fact that the presence of quantum coherence in a suitably chosen reference basis is equivalent to the nonpositivity of the extended KD quasiprobability distribution. After the application of the $Y$-twirling operation, this establishes a direct operational link between imaginarity and KD nonpositivity. Finally, we formulate a hierarchy of sufficient conditions in terms of moments of the KD distribution, enabling efficient detection without requiring full quasiprobability reconstruction.

\subsection{Removal of real coherences by the \texorpdfstring{$Y$}{Y}-twirl map}
In this subsection, we consider a quantum operation that eliminates the real part of the coherence while preserving its imaginary component, referred to as the $Y$-twirl operation.
\begin{lemma}
\label{thm:Ytwirl}
Let $\mathcal{H}$ be a $d$-dimensional Hilbert space with orthonormal basis
$\{\ket{n}\}_{n=0}^{d-1}$.  
For each pair of distinct indices $p<q$, define the antisymmetric operator
\begin{align}
\mathcal{Y}_{pq} := i\big(\ket{p}\bra{q} - \ket{q}\bra{p}\big), \label{eq:Ytwirl}
\end{align}
which acts non-trivially only on the two-dimensional subspace spanned by $\{\ket{p},\ket{q}\}$.
Then, define the linear Y-twirl map $\mathcal{E}:\mathcal{B}(\mathcal{H}) \to \mathcal{B}(\mathcal{H})$ by
\begin{equation}
\mathcal{E}(\rho)
= \frac{1}{d}\left(\rho + \sum_{p<q} \mathcal{Y}_{pq}\,\rho\,\mathcal{Y}_{pq}\right).
\label{eq:Ytwirl-map}
\end{equation}
For any density operator $\rho$, with matrix elements $\rho_{mn}=\bra{m}\rho\ket{n}$,
the output state $\rho'=\mathcal{E}(\rho)$ satisfies:
\begin{align}
\rho'_{mm} &= \frac{1}{d}\
&&\text{for all diagonal entries,}
\label{eq:Ytwirl-diagonal}\\[4pt]
\rho'_{mn} &= \frac{2}{d}\, i\, \operatorname{Im}(\rho_{mn})
&&\text{for all } m\neq n.
\label{eq:Ytwirl-offdiag}
\end{align}
That is, the $Y$-twirl map completely removes the real parts of all off-diagonal elements,
keeping only the imaginary components up to a factor $2/d$, while forcing all diagonal entries to be equal.
\end{lemma}
\begin{proof}
Let $\rho=\sum_{m,n}\rho_{mn}\ket{m}\bra{n}$ be an arbitrary density operator on a $d$-dimensional Hilbert space. Since the map $\mathcal E$ is linear, it is sufficient to determine its action on the matrix units $\ket{m}\bra{n}$.

\textbf{Diagonal terms:} Let $m=n$, then
\[
\mathcal E(\ket{m}\bra{m})
=
\frac{1}{d}\left(
\ket{m}\bra{m}
+
\sum_{p<q} \mathcal Y_{pq}\ket{m}\bra{m}\mathcal Y_{pq}
\right).
\]
If $m\notin\{p,q\}$, then $\mathcal Y_{pq}\ket{m}=0$ and the contribution vanishes. If $m\in\{p,q\}$, say $(p,q)=(m,r)$, then
\[
\mathcal Y_{mr}\ket{m}=-i\ket{r},
\qquad
\bra{m}\mathcal Y_{mr}=i\bra{r},
\]
and therefore
\[
\mathcal Y_{mr}\ket{m}\bra{m}\mathcal Y_{mr}
=
\ket{r}\bra{r}.
\]
Summing over all $r\neq m$ gives
\[
\sum_{p<q}\mathcal Y_{pq}\ket{m}\bra{m}\mathcal Y_{pq}
=
\sum_{r\neq m}\ket{r}\bra{r}.
\]
Hence
\begin{align}
\mathcal E(\ket{m}\bra{m})
&= \frac{1}{d} \left( \ket{m}\bra{m} + \sum_{r\neq m} \ket{r}\bra{r} \right) \nonumber \\
&= \frac{1}{d} \sum_{r =1}^d \ket{r}\bra{r}
\end{align}
which implies
\[
\rho'_{mm}=\frac{1}{d}.
\]

\textbf{Off-diagonal terms:} Now let $m\neq n$, then
\[
\mathcal E(\ket{m}\bra{n})
=
\frac{1}{d}\left(
\ket{m}\bra{n}
+
\sum_{p<q}\mathcal Y_{pq}\ket{m}\bra{n}\mathcal Y_{pq}
\right).
\]

If $\{p,q\}$ is disjoint from $\{m,n\}$, the contribution vanishes. If $(p,q)=(m,n)$, then
\[
\mathcal Y_{mn}\ket{m}=-i\ket{n},\qquad
\bra{n}\mathcal Y_{mn}=i\bra{m},
\]
and therefore
\[
\mathcal Y_{mn}\ket{m}\bra{n}\mathcal Y_{mn}
=
-\ket{n}\bra{m}.
\]

If exactly one index overlaps, say $(p,q)=(m,r)$ with $r\neq n$, then
\[
\mathcal Y_{mr}\ket{m}=-i\ket{r},\qquad
\bra{n}\mathcal Y_{mr}=0,
\]
so the contribution vanishes. The same holds when $(p,q)=(n,r)$ with $r\neq m$.

Thus the only nonzero contribution comes from $(p,q)=(m,n)$, yielding
\[
\sum_{p<q}\mathcal Y_{pq}\ket{m}\bra{n}\mathcal Y_{pq}
=
-\ket{n}\bra{m}.
\]

Hence
\[
\mathcal E(\ket{m}\bra{n})
=
\frac{1}{d}
\left(
\ket{m}\bra{n}
-
\ket{n}\bra{m}
\right).
\]

For the density matrix $\rho=\sum_{m,n}\rho_{mn}\ket{m}\bra{n}$,
we therefore obtain
\[
\rho'_{mn}
=
\frac{1}{d}(\rho_{mn}-\rho_{nm})
=
\frac{2}{d}\,i\,\mathrm{Im}(\rho_{mn}),
\qquad \forall \, m\neq n.
\]

This completes the proof.
\end{proof}

The above lemma demonstrates that the $Y$-twirl map singles out the intrinsically nonclassical contribution stemming from imaginarity.
Below, we explicitly verify that the $Y$-twirl map $(\mathcal{E})$ reproduces the expected form for Hilbert spaces of dimension $d=2$ and $d=3$.

\subsubsection*{Demonstration for qubits and qutrits}

\textbf{Qubit case ( $d=2$ ) :} In this case, there is only one antisymmetric operator,
\[
\mathcal{Y}_{01} = i(\ket{0}\bra{1} - \ket{1}\bra{0}).
\]
Consider a general qubit density matrix
\[
\rho =
\begin{pmatrix}
a & b\\[2pt]
b^* & 1-a
\end{pmatrix},
\]
where $0\leq a\leq 1$ and $|b|^2 \leq a(1-a)$, the latter condition ensuring the positive semidefiniteness of $\rho$. After applying the $Y$ twirling map, we obtain
\[
\mathcal{E}(\rho)
= \tfrac{1}{2}\big(\rho + \mathcal{Y}_{01}\rho\,\mathcal{Y}_{01}\big)
= 
\begin{pmatrix}
\tfrac{1}{2} & i\,\Im(b)\\[3pt]
-\,i\,\Im(b) & \tfrac{1}{2}
\end{pmatrix}.
\]
This clearly implies that the diagonal elements are uniformly averaged to $1/2$, the real part of the off-diagonal term is eliminated, and only the imaginary component is retained.

\textbf{Qutrit case ( $d=3$ ) :} For $d=3$, the antisymmetric operators are
\begin{align*}
\mathcal{Y}_{01} = i(\ket{0}\bra{1} - \ket{1}\bra{0}), \\
\mathcal{Y}_{02} = i(\ket{0}\bra{2} - \ket{2}\bra{0}), \\
\mathcal{Y}_{12} = i(\ket{1}\bra{2} - \ket{2}\bra{1}).
\end{align*}

We consider a general qutrit density matrix
\[
\rho =
\begin{pmatrix}
a & b & c\\[2pt]
b^* & d & e\\[2pt]
c^* & e^* & f
\end{pmatrix},
\quad \text{where}~~ a,d,f \geq 0 \; \text{and }a+d+f=1.
\]
The conditions for $\rho$ to be a valid density matrix are
\begin{align}
a,d,f \ge 0, \qquad a+d+f = 1, \nonumber \\
|b|^2 \le ad, \qquad |c|^2 \le af, \qquad |e|^2 \le df, \nonumber\\
adf + 2\,\mathrm{Re}(bce^*) \ge a|e|^2 + d|c|^2 + f|b|^2. \label{positivity_rho}
\end{align}
Next we apply the $Y$ twirl map,
\begin{align*}
\mathcal{E}(\rho)
&= \tfrac{1}{3}\!\left(\rho
+ \mathcal{Y}_{01}\rho\mathcal{Y}_{01}
+ \mathcal{Y}_{02}\rho\mathcal{Y}_{02}
+ \mathcal{Y}_{12}\rho\mathcal{Y}_{12}\right)\\
&= \begin{pmatrix}
\tfrac{1}{3} & \tfrac{2i}{3}\Im(b) & \tfrac{2i}{3}\Im(c)\\[4pt]
-\tfrac{2i}{3}\Im(b) & \tfrac{1}{3} & \tfrac{2i}{3}\Im(e)\\[4pt]
-\tfrac{2i}{3}\Im(c) & -\tfrac{2i}{3}\Im(e) & \tfrac{1}{3}
\end{pmatrix}.
\end{align*}
The diagonal entries are uniform ($1/3$), the real parts of all off-diagonal elements vanish, and the imaginary parts are rescaled by the universal factor $2/d=2/3$.

\subsection{Equivalence between quantum coherence and nonpositivity of the extended KD distribution}\label{s3A}
Now, we exploit a correspondence between quantum coherence in a fixed reference basis and the nonpositivity of the associated extended KD distribution constructed using two mutually unbiased bases.

\begin{lemma} \cite{budiyono2023quantifying}
\label{thm:KD_coherence_equivalence}
Let $\{\ket{a_i}\}$ and $\{\ket{b_k}\}$ be two orthonormal bases of a
$d$-dimensional Hilbert space, with $\{\ket{b_k}\}$ mutually unbiased with
respect to $\{\ket{a_i}\}$. For any quantum state $\rho$, the extended KD distribution $\{ Q^{\star}_{j,k,l}(\rho) \}$ as defined in Eq.(\ref{eq:extended_KD_definition}), satisfies
\begin{equation}
\sum_{j,k,l} \big| Q^{\star}_{j,k,l}(\rho) \big| - 1
=
\mathcal{C}_{\ell_1}(\rho,\{\ket{a_j}\}),
\end{equation}
where $\mathcal{C}_{\ell_1}(\rho,\{\ket{a_j}\})$ denotes the $\ell_1$ norm of coherence of $\rho$ in the basis $\{\ket{a_j}\}$. Consequently, the extended KD distribution is nonpositive, i.e.,
\begin{equation}
\sum_{j,k,l} \big| Q^{\star}_{j,k,l}(\rho) \big| - 1 > 0,
\end{equation}
if and only if $\rho$ possesses nonzero coherence in the basis $\{\ket{a_j}\}$.
\end{lemma}

\begin{proof}
Recall that the $\ell_1$ norm of coherence of a state $\rho$ with respect to the reference basis $\{\ket{a_j}\}$ is defined as~\cite{baumgratz2014quantifying}:
\begin{equation}
\mathcal{C}_{\ell_1}(\rho,\{\ket{a_j}\})
= \sum_{j\neq k} \big| \langle a_j | \rho | a_k \rangle \big|.
\end{equation}
Using the normalization condition $\sum_j \langle a_j | \rho | a_j \rangle = 1$, this expression can be rewritten as
\begin{equation}
\mathcal{C}_{\ell_1}(\rho,\{\ket{a_j}\})
=
\sum_{j,k} \big| \langle a_j | \rho | a_k \rangle \big| - 1.
\end{equation}

Since the two bases are mutually unbiased, their overlaps satisfy
$\big| \langle a_j | b_l \rangle \big| = {1}/{\sqrt{d}},\;
\; \forall j,l.$
It implies
\begin{equation}
\big| \langle a_k | b_l \rangle \langle b_l | a_j \rangle \big| = \frac{1}{d} \quad \forall j,k,l.
\end{equation}
Therefore,
\begin{align}
\sum_{j,k,l}
\big| Q^{\star}_{j,k,l}(\rho) \big| -1 
&=
\sum_{j,k,l} \big| \langle a_k | b_l \rangle \langle b_l | a_j \rangle \langle a_j | \rho | a_k \rangle \big| -1 \nonumber \\
&= \sum_{j,k} \big| \langle a_j | \rho | a_k \rangle \big| - 1 \nonumber\\
& = \mathcal{C}_{\ell_1}(\rho,\{\ket{a_j}\}).
\end{align}
This establishes the desired equivalence. The final statement follows immediately: $\mathcal{C}_{\ell_1}(\rho,\{\ket{a_j}\})$ vanishes if and only if the state is diagonal in the basis $\{\ket{a_j}\}$, which occurs when the extended KD distribution is positive.
\end{proof}

The above result shows that extended KD nonpositivity provides a faithful and quantitative witness of quantum coherence. In particular, the excess absolute weight of the extended KD distribution above unity is exactly equal to the $\ell_1$ norm of coherence, establishing a direct one-to-one correspondence between coherence and the extended KD nonpositivity in the mutually unbiased bases. Combining this correspondence with the 
$Y$-twirl operation, we obtain a direct link between coherence arising solely from imaginarity and the nonpositivity of the extended KD distribution. In the next subsection, we now move towards a moment-based framework for detecting quantum imaginarity.

\subsection{KD moment-based detection of quantum imaginarity}

To demonstrate our moment-based framework for detecting quantum imaginarity, we first define the moments of the extended KD distribution.

\begin{definition}
 Let $\mathcal{H}$ be a $d$-dimensional Hilbert space and ${\rho}$ be an arbitrary density operator defined on $\mathcal{H}$. Consider $\{ Q^{\star}_{ijk}(\rho) \}$ to be the extended KD distribution for the quantum state $\rho$ with respect to bases $\{\ket{a_i}\}$ and $\{\ket{b_k}\}$, where $\{\ket{b_k}\}$ is mutually unbiased with respect to $\{\ket{a_i}\}$. We define the $n$-th order extended KD moments $(r_n) $ as:

\begin{equation} 
 r_n := \sum_{i,j,k}
 ( Q^{\star}_{ijk} (\rho) )^n\label{extendedkdmoments}
\end{equation}
where $n$ is an integer.
\end{definition}

Note that, for a fixed reference basis $\{\ket{a_i}\}$ in a $d$-dimensional Hilbert space, one can always choose a mutually unbiased basis $\{\ket{b_k}\}$ with appropriate phase conventions such that the corresponding extended KD distribution becomes real. We now formulate a refined criterion for detecting KD nonpositivity using moments of the extended KD distribution.

\begin{theorem} \label{theorem1}
If the extended KD distribution $\{ Q^{\star}_{ijk} (\rho)\}$ for a quantum state $\rho$ with respect to bases $\{\ket{a_i}\}$ and $\{\ket{b_k}\}$ is positive, where $\{\ket{b_k}\}$ is a mutually unbiased basis with respect to $\{\ket{a_i}\}$, then 
\begin{equation}
\det[H_{m}(\mathbf{r})] \ge 0 . \label{Hankelmatrix_coherence}
\end{equation}
Here, $[H_{m}(\mathbf{r})]_{pq} = r_{p+q+1}$ for $p,q \in \{0,1,\dots,m\}$, $m \in \mathbb{N}$ and $ \mathbf{r}= (r_1, r_2, \dots, r_{2m+1})$ are the set of corresponding extended KD moments defined in Eq.~\eqref{extendedkdmoments}.
\end{theorem}

\begin{proof} 
If the extended KD distribution is positive, its elements form a normalized collection of nonnegative real numbers $\{\lambda_\ell \ge 0\}$. Placing these elements on the diagonal yields a matrix $D$ that is positive semi-definite and normalized to unit trace. The extended KD moments arise from a positive distribution, implying that the associated Hankel matrices admit a Vandermonde decomposition of the form $H_m(\mathbf r)=\mathcal{V}_m D \mathcal{V}_m^T$, where $\mathcal{V}_m$ is the Vandermonde matrix with entries $(\mathcal{V}_m)_{j\ell}=\lambda_\ell^{\,j-1}$. It follows immediately that $H_m(\mathbf r)$ is positive semi-definite, and hence $\det[H_m(\mathbf r)] \ge 0$. The detailed proof of Theorem \ref{theorem1} is provided in Appendix \ref{appendixA}.
\end{proof}

Theorem~\ref{theorem1} asserts that the condition stated in Eq.~\eqref{Hankelmatrix_coherence} is \emph{necessary} for the extended KD distribution to remain positive. As a result, any violation of this condition serves as a \emph{sufficient} indicator of nonpositivity. For \(m > 1\), Theorem~\ref{theorem1} introduces a hierarchy of increasingly stronger criteria that can detect nonpositivity of the extended KD distribution with higher sensitivity.

We now aim to detect quantum imaginarity through the above established moment-based KD nonpositivity detection scheme, which is formally stated as the following theorem.

\begin{theorem}\label{theorem2}
Suppose that $\rho' = \mathcal{E}(\rho)$ is obtained by applying the $Y$-twirl operation, defined in Eq.~\eqref{eq:Ytwirl}, to the state $\rho$. Let $Q^{\star}(\rho')$ be the extended KD distribution associated with a quantum state $\rho'$, defined with respect to two mutually unbiased bases $\{\ket{a_i}\}$ and $\{\ket{b_k}\}$. If there exists an integer $m \in \mathbb{N}$ such that
\begin{equation}
 \det[H_{m}(\mathbf{r}{'})] < 0,
\label{moment_coherence}
\end{equation}
then the state $\rho$ possesses nonzero imaginarity in the basis $\{\ket{a_i}\}$. Here, $ \mathbf{r}{'}= (r'_1, r'_2, \dots, r'_{2m+1})$ and $r'_{n} = \sum_{i,j,k} \left[ Q^{\star}_{ijk}\, (\rho') \right]^n$ denote the extended KD moments associated with the state $\rho'$.
\end{theorem} 
\begin{proof} 
The proof of this theorem is constructive.

If the inequality $ \det[H_{m}(\mathbf{r}{'})] < 0$ holds for some $m \in \mathbb{N}$, then by Theorem~\ref{theorem1}, the extended KD distribution $Q^{\star}(\rho')$ must be nonpositive. From Lemma~\ref{thm:KD_coherence_equivalence}, nonpositivity of the extended KD distribution is equivalent to nonzero $\ell_1$-norm coherence in the reference basis $\{\ket{a_i}\}$. Therefore, the state $\rho'$ must exhibit nonzero coherence in that basis.
Next, recall from Lemma~\ref{thm:Ytwirl} that the $Y$-twirl operation removes all real components of coherence while preserving the imaginary ones. Consequently, any coherence present in $\rho' = \mathcal{E}(\rho)$ must be purely imaginary. Hence, if $\rho'$ has a nonzero coherence component, that coherence is necessarily imaginary. Since $\rho'$ is obtained from $\rho$ via the $Y$-twirl operation, the presence of imaginary coherence in $\rho'$ implies that the original state $\rho$ also possessed imaginary coherence. Therefore, $\rho$ has nonzero imaginarity in the basis $\{\ket{a_i}\}$. This completes our proof.
\end{proof}

Thus, Theorems \ref{theorem1} and \ref{theorem2} establish a hierarchy of moment-based \emph{sufficient} conditions for detecting the nonpositivity of the extended KD distribution and nonzero imaginarity respectively. These criteria are primarily motivated by experimental considerations, as the relevant moments can be experimentally realized (see Sec.~\ref{s4} and Appendix~\ref{app:shadow} for further details).

We now present two examples validating Theorem~\ref{theorem2}, demonstrating how quantum imaginarity can be detected using our proposed moment criteria.

\subsection{Illustrative qubit example}
\label{sec:example_imaginarity}

We consider the pure single-qubit state
\begin{equation}
    \ket{\Psi}
    =
    \cos\!\left(\frac{\theta}{2}\right)\ket{0}
    +
    \sin\!\left(\frac{\theta}{2}\right)e^{i\alpha}\ket{1}, \label{qubit_example}
\end{equation}
where $0 \le \theta \le \pi$ and $0 \le \alpha \le 2\pi$. The amount of imaginarity of this state $\rho = \ket{\Psi} \bra{\Psi}$ is given by,
\begin{align}
    \mathcal{M}_{\ell_1} (\rho) = \sum_{i\neq j} \big| \text{Im} (\rho_{ij}) \big| = \big| \sin {\alpha} \sin \theta \big| 
\end{align}
Now we apply the $Y$-twirl operation,
\begin{equation}
    \rho'=\mathcal{E}(\rho)=\frac12(\rho+\mathcal{Y}_{01}\rho\,\mathcal{Y}_{01}).
\end{equation}
In this example, we consider  $\theta=\pi/2$, such that $M_{\ell_1} (\rho) = |\sin \alpha|$. The resulting state after the Y-twirl operation is given by,
\begin{equation}
\rho'=
\begin{pmatrix}
\frac12 & -\frac{i}{2}\sin\alpha \\
\frac{i}{2}\sin\alpha & \frac12
\end{pmatrix}.
\end{equation}
Imaginarity is therefore present whenever $\sin\alpha\neq 0$.

We now construct the extended KD distribution using the mutually unbiased basis
\begin{align}
    \ket{b_1} &= \frac{1}{\sqrt{2}}(\ket{0}+e^{i\beta}\ket{1}), \nonumber\\
    \ket{b_2} &= \frac{1}{\sqrt{2}}(\ket{0}-e^{i\beta}\ket{1}), \label{mutulally_unbiased_bases}
\end{align}
where $0 \leq \beta \leq 2\pi$. Although imaginarity in the basis \( \{ \ket{a_i} \} \) can be detected via the nonpositivity of the extended KD distribution using a mutually unbiased basis \( \{ \ket{b_k} \} \), it is important to note that the extended KD distribution itself depends on the choice of \( \{ \ket{b_k} \} \). Consequently, different parameter choices ({\it e.g.}, values of \( \beta \)) may require the evaluation of determinants of Hankel matrices of different orders to detect imaginarity using our moment-based approach. In Figure~\ref{fig:imaginarity_hierarchy}, we consider three different values of \( \beta \) and plot the negative of the determinant corresponding to the {\it minimum} order of the Hankel matrix required to detect imaginarity across the full range of \( \alpha \). The detailed calculations for this example can be found in Appendix~\ref{app:kd_calculations}.

\begin{figure}[t]
\centering

\includegraphics[width=1\linewidth]{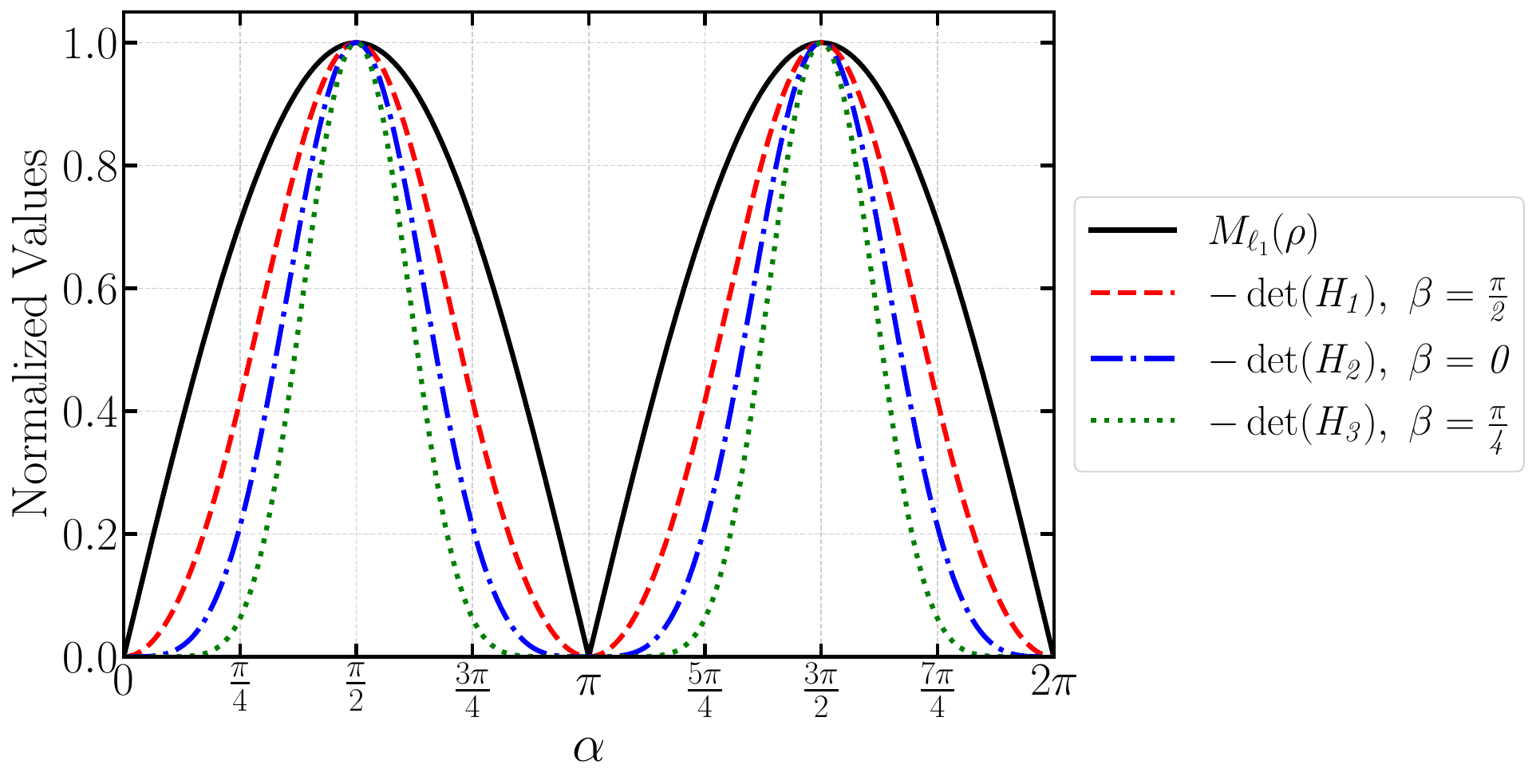}
\caption{ 
\justifying 
\small
{Amount of imaginarity for the state $\rho$ with $\theta= \pi/2$, $(\mathcal{M}_{\ell_1 }{(\rho)}=|\sin\alpha|)$ (black) and the minimal-order Hankel determinants detecting the entire region with non-zero imaginarity are plotted as a function of $\alpha$. For $\beta=\pi/2$, detection occurs at first order ($H_1$). For $\beta=0$, detection requires $H_2$. For $\beta=\pi/4$, first and second orders fail, and detection occurs at $H_3$.
All curves vanish only at $\alpha=0,\pi,2\pi$.}
}
\label{fig:imaginarity_hierarchy}
\end{figure}


\subsection{Higher-dimensional illustration: phase-GHZ states}

To illustrate the applicability of the proposed framework beyond two-dimensional systems, we consider the family of $n$-qubit phase-GHZ states
\begin{equation}
\ket{\psi_n}
=
\frac{\ket{0}^{\otimes n}+i\ket{1}^{\otimes n}}{\sqrt2},
\end{equation}
belonging to a Hilbert space of dimension $D=2^n$.

Applying the $Y$-twirl map of Lemma~\ref{thm:Ytwirl}, the resulting state takes the form
\begin{equation}
\rho'
=
\frac{\mathbb{I}}{D}
+
\frac{i}{D}
\left(
\ket{D-1}\!\bra{0}
-
\ket{0}\!\bra{D-1}
\right).
\label{eq:GHZ_Ytwirl}
\end{equation}
Here we choose the computational basis
\(
\mathcal A=\{\ket{a_i}\}_{i=0}^{D-1}
\)
together with the tensor-product mutually unbiased basis
\begin{equation}
\mathcal B
=
\{
\ket{y_{\pm}}
\otimes
\ket{x_{\pm}}^{\otimes(n-1)}
\},
\end{equation}
where
\[
\ket{x_{\pm}}
=
\frac{\ket0\pm\ket1}{\sqrt2},
\qquad
\ket{y_{\pm}}
=
\frac{\ket0\pm i\ket1}{\sqrt2}.
\]
By construction, every basis vector of $\mathcal B$ is mutually unbiased with respect to the computational basis.

Next, we calculate the extended KD elements using Eq.~(\ref{eq:extended_KD_definition}). The corresponding extended KD distribution assumes only three values,
\begin{equation}
Q^\star_{ijk}
\in
\left\{
-\frac1{D^2},
\,0,\,
\frac1{D^2}
\right\},
\end{equation}
with $D(D+1)$ positive entries, $D$ negative entries, and all remaining entries equal to zero. Consequently, the $m$-th extended KD moment is
\begin{align}
r_m
&=
D(D+1)\left(\frac1{D^2}\right)^m
+
D\left(-\frac1{D^2}\right)^m \nonumber\\
&=
\frac{D+1+(-1)^m}{D^{2m-1}}.
\label{eq:GHZmoments}
\end{align}
In particular, $r_2={(D+2)}/{D^3}$, $r_3=1/{D^4}$.
Therefore,
\begin{equation}
\det H_1(\mathbf r')
=
r_1r_3-r_2^2
=
-\frac{4(D+1)}{D^6},
\end{equation}
which is negative for every finite dimension $D$. Hence, the first Hankel condition certifies nonzero imaginarity for the entire family of phase-GHZ states.

\section{Interferometric Scheme for the realization of moments and imaginarity detection} \label{s4}
In this section, we present a proposal to realize the moments from the visibility of an interference pattern. Furthermore, while moment-based criteria provide a sufficient method for detecting imaginarity, here we also introduce an alternative approach that enables the direct detection of imaginarity from the visibility obtained in a Mach–Zehnder interferometer.

\subsection{Realization of moments in an interferometric setup} \label{momentrealization}
The moments $r'_n$ used to detect imaginarity can be readily implemented in an experimental setup. In support of this claim, here we propose a scheme for the realization of these moments from the visibility of an interference pattern. For this purpose, we use the extended KD moments $r'_{n}$ defined on the state $\rho'=\mathcal{E}(\rho)$, where $\mathcal{E}$ is the linear Y-twirl map given by Eq.~\eqref{eq:Ytwirl-map}. The moments can be represented as 
\begin{align}
    r'_n &= \sum_{i,j,k} \left[Q^{\star}_{ijk} (\rho')\right]^n \nonumber \\
    &= \sum_{i,j,k} \left[ \mathrm{Tr}(\Pi_k^b \Pi_j^a \Pi_i^a \, \rho' ) \right]^n \nonumber (\text{Using Eq.~\eqref{eq:extended_KD_definition}})\\
     &= \sum_{i,j,k} \mathrm{Tr}\left[\left(\Pi_k^b \Pi_j^a \Pi_i^a \right)^{\otimes n} \, (\rho')^ {\otimes n}\right] \nonumber \\
     &= \mathrm{Tr} \left[ \left( \sum_{i,j,k} \left(\Pi_k^b \Pi_j^a \Pi_i^a \right)^{\otimes n} \right)  (\rho')^{\otimes n} \right] \nonumber \\
     & = \mathrm{Tr}\left[ S_n (\rho')^{\otimes n} \right] \;,\label{nonlinear_KD}
\end{align}
\text{where } $S_n = \sum_{i,j,k} \left(\Pi_k^{b} \Pi_j^a  \Pi_i^a \right)^{\otimes n},$ and $\Pi_k^{b}, \Pi_j^a$ and $\Pi_i^a$ are the projectors associated with the measurement bases $\{\ket{b_k}\}$, $\{\ket{a_j} \}$ and $\{ \ket{a_i} \}$, respectively. 
Thus, every extended KD moment $(r'_n)$ admits a representation as the expectation value of a multicopy operator $(S_n)$ acting on $n$ identical copies of the quantum state. Consequently, the proposed moments fall within the general class of nonlinear multicopy functions considered in the classical-shadow framework~\cite{huang2020predicting}, making them directly amenable to estimation using existing randomized-measurement protocols. For a more detailed discussion on shadow estimation, see Appendix~\ref{app:shadow}.

It should be noted that the operator $S_n$ defined above is unitary $\forall$ $n \in \mathbb{N}$. This is evident from the following:
\begin{align*}
\Pi_k^{b}\Pi_j^{a}\Pi_i^{a}
&= \ket{b_k}\bra{b_k}\ket{a_j}\bra{a_j}\ket{a_i}\bra{a_i} \\
&= \langle b_k|a_j\rangle \langle a_j|a_i\rangle \ket{b_k}\bra{a_i} \\
&= \langle b_k|a_i\rangle \ket{b_k}\bra{a_i}.
\end{align*}
Substituting this into the definition of $S_n$ gives
\[
S_n=\sum_{i,k}\left(\langle b_k|a_i\rangle\,\ket{b_k}\bra{a_i}\right)^{\otimes n}.
\]
Taking the adjoint,
\[
S_n^\dagger=\sum_{i,k}\left(\langle a_i|b_k\rangle\,\ket{a_i}\bra{b_k}\right)^{\otimes n}.
\]
Using the orthonormality relations $\langle b_k|b_l\rangle=\delta_{kl}$ and the completeness relation
$\sum_k\ket{b_k}\bra{b_k}=\mathbb{I}$, we obtain
\begin{align*}
S_n^\dagger S_n
&=\sum_{i,j,k}\left(\langle a_i|b_k\rangle\langle b_k|a_j\rangle\,\ket{a_i}\bra{a_j}\right)^{\otimes n}\\
&=\sum_{i}\left(\ket{a_i}\bra{a_i}\right)^{\otimes n}\\
&=\mathbb{I}^{\otimes n} \; ,
\end{align*}
which is the identity operator on the $n$-copy Hilbert space $\mathcal{H}^{\otimes n}$. Hence $S_n$ is unitary.

Below, we provide an interferometric scheme for the realization of the second and third order moments. The realization of all higher-order moments follows similarly.

\textbf{\textit{Realization of the second order moment:}} The second order moment is given by 
\begin{equation}
\begin{split}
    r'_2=& \sum_{i,j,k} \left[Q^{\star}_{ijk} (\rho')\right]^2 \\ & =\mathrm{Tr}\left[ S_2 (\rho')^{\otimes 2} \right].
    \end{split}
\end{equation}
Now, consider a Mach-Zehnder interferometer with an input $(\rho')^{\otimes 2}$ corresponding to the internal degree of freedom. For finding the second-order moment, we choose the unitary to be $S_2$. For this choice of input and unitary, the visibility is given by 
\begin{equation}
\begin{split}
    V^{(2)}=& |\mathrm{Tr}\left[ S_2 (\rho')^{\otimes 2} \right]| (\text{following Eq.~\eqref{visibility1}})\\ & = |r'_2| \\ & =r'_2, 
    \end{split}
\end{equation}
So the second-order moment is directly obtained from the visibility $V^{(2)}$ of an interference pattern corresponding to this choice of input and unitary.

\textbf{\textit{Realization of the third order moment:}} This is given by 
\begin{equation}
\begin{split}
    r'_3=& \sum_{i,j,k} \left[Q^{\star}_{ijk} (\rho')\right]^3 \\ & =\mathrm{Tr}\left[ S_3 (\rho')^{\otimes 3} \right].
    \end{split}
\end{equation} 
By choosing the unitary to be $S_3$, and the input state to be $(\rho')^{\otimes 3}$, the visibility reduces to 
\begin{equation}
\begin{split}
    V^{(3)}=& |\mathrm{Tr}\left[ S_3 (\rho')^{\otimes 3} \right]| \, (\text{following Eq.~\eqref{visibility1}})\\ & = |r'_3| \\ & = r'_3. 
    \end{split}
\end{equation}
The last line follows from the normalization of the extended KD distribution, which ensures that the total positive contribution outweighs any negative contribution. Consequently, $r_3'$ remains positive.

By generalizing this, one can obtain the $n$th-order moment for an arbitrary $n \in \mathbb{N}$. Specifically, by choosing the input and unitary to be  $(\rho')^{\otimes n}$ and $S_n$ respectively, the corresponding moment is obtained from the visibility $V^{(n)}$. This is illustrated in Figure~\ref{fig:inteferometer_2}. Therefore, this provides a practical scheme for realizing the moments, enabling an experimentally feasible method for detecting quantum imaginarity. Note that, in addition to the moment-based detection method, the imaginarity of an unknown quantum state can also be detected directly from the visibility of a Mach–Zehnder interferometer, as discussed in the following section.

\begin{figure}[t]
\centering
\includegraphics[width=1\linewidth]{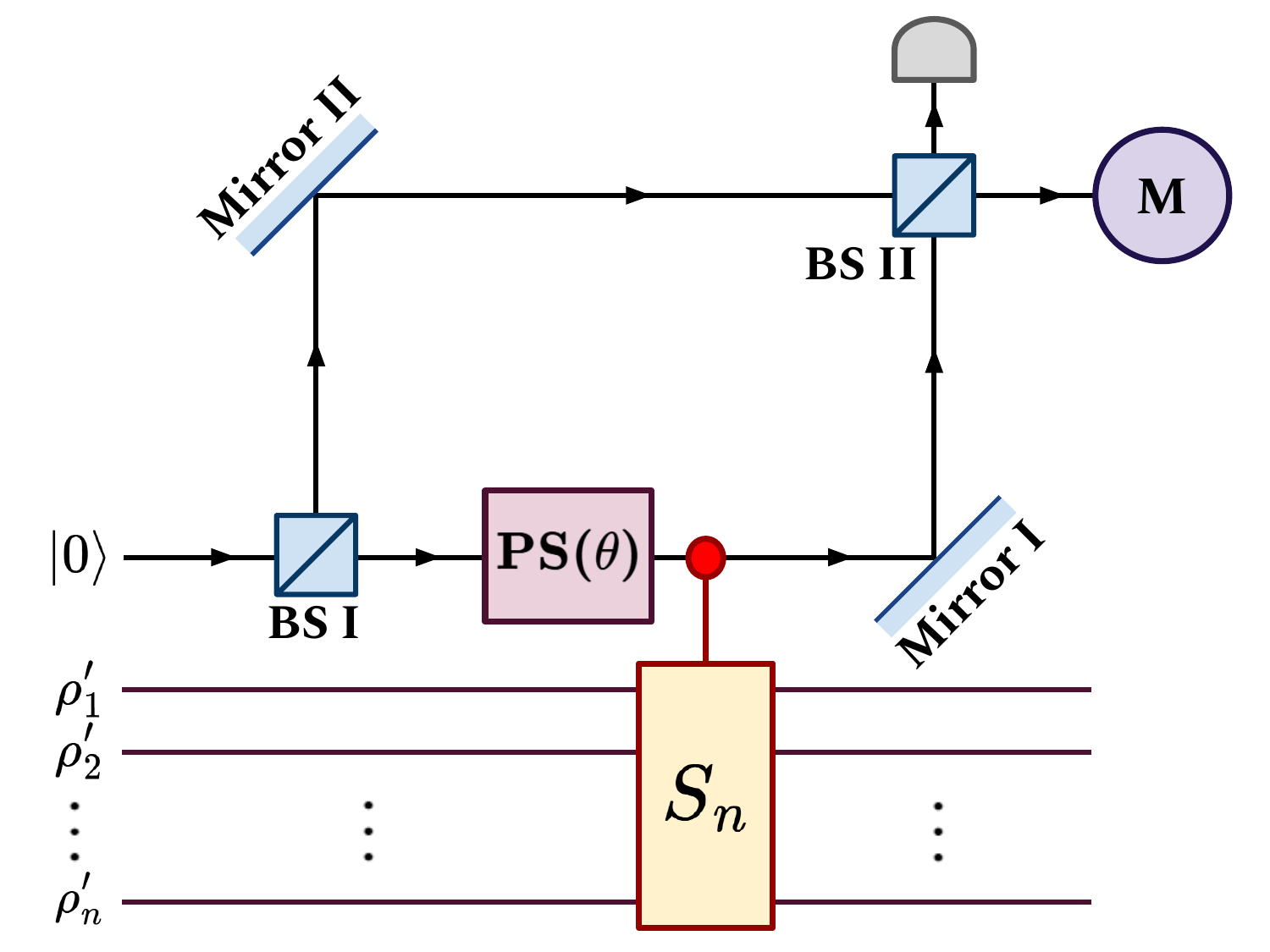}
\caption{ 
\justifying 
\small{
Generalized Mach–Zehnder interferometric scheme for measuring the $n$th-order extended KD moment. The input state is $\rho'_1 \otimes\rho'_2 \otimes \dots \otimes \rho'_n$, which represents $n$ copies of the Y-twirled state $\rho'$. The controlled operation implements the unitary $S_n$ acting on the internal degrees of freedom. The visibility of the resulting interference pattern yields $V^{(n)}$, which corresponds to the $n$th-order moment. The remaining terminologies follow the same convention as in Figure~\ref{fig:inteferometer}.
}}
\label{fig:inteferometer_2}
\end{figure}

\subsection{Direct detection of imaginarity from interferometric visibility}
Quantum imaginarity is another signature of nonclassicality, akin to visibility. In what follows, we establish the connection between these two quantities and present a theorem, which enables the direct detection of imaginarity through the visibility measured in an interferometric setup. 

\begin{theorem}\label{theorem3}
Let $\rho \in \mathcal{D}(\mathbb{C}^d)$ be an arbitrary quantum state, and let $\{G_j^{I}\}$ denote the set of imaginary operators (operators having non-zero imaginary matrix elements) in the $d$-dimensional generalized Pauli basis. For each $G_j^{I}$, define the unitary operator 
\begin{equation}
U_j^I (\theta)= \exp{\left( {i \theta G_j^{I}} \right)}, \qquad \text{where} \; \theta \in \mathbb{R}.
\end{equation}
The corresponding interferometric visibility for this unitary operator (as defined in Eq.~\eqref{visibility1}) is
\begin{equation}
V_j (\theta) = \left| \mathrm{Tr}\left[ U_j^I (\theta)\,\rho\right] \right|,
\end{equation}
and the total visibility is given by $V (\theta)= \sum_j V_j (\theta)$. Then $V \left(\theta = \frac{\pi}{2}\right)$ is equal to the $\ell_1$-norm measure of imaginarity, i.e., 
\begin{equation}
    V \left(\theta = \frac{\pi}{2}\right)= \mathcal{M}_{\ell_1}. 
\end{equation}
Hence, $V \left(\theta = \frac{\pi}{2}\right) $ vanishes if and only if the quantum state $\rho$ has zero imaginarity.
\end{theorem}
\begin{proof}
Any quantum state $\rho \in \mathcal{D}(\mathbb{C}^d)$ can be expanded in a generalized Pauli basis as
\begin{equation}
\rho = \sum_j b_j G_j \;,
\end{equation}
where $b_j \in \mathbb{R}$ and $\{G_j\}$ are Hermitian operators satisfying $\mathrm{Tr}(G_j G_k) = 2\,\delta_{jk}$.

This basis can be decomposed into real $\{G_j^{R}\}$ and imaginary $\{G_{j}^{I}\}$ parts as,
\begin{equation}
    \rho=\sum_i b_i^R G_i^{R}+\sum_j b_j^I G_j^{I}
\end{equation}
Note that, $\{G_j^{I}\}$ is exactly the set of complex antisymmetric generators $\{Y_{pq}\}$ defined in Eq.~\eqref{eq:antisymmetric_generators}, i.e.,
\begin{equation}
G_{j}^{I} \equiv Y_{pq} = i\big(\ket{a_p}\bra{a_q} - \ket{a_q}\bra{a_p}\big), \quad p<q,
\end{equation}
where the antisymmetric generators are written in the basis $\{ \ket{a_i} \}_{i=0}^{d-1}$. The coefficients of these antisymmetric generators $\{ b_{j}^I \}$ characterize the imaginarity of the state. The $\ell_1$-norm measure of imaginarity for this state is given by \cite{Chen2023},
\begin{equation}
    \mathcal{M}_{\ell_1} (\rho) = \sum_{j\neq k} |\operatorname{Im} (\rho_{jk})| = 2 \sum_{j}|b_j^I|
\end{equation}

For $d=2$, the operators $Y_{pq}$ are Hermitian as well as unitary, however for $d>2$, the operators $Y_{pq}$ are Hermitian, but not necessarily unitary. Therefore, to access them experimentally in the interferometric setup, we consider the unitary operator $U_{j}^I$ generated by $G_j^I (= Y_{pq})$,
\begin{equation*}
U_{j}^I(\theta) = \exp \left( i\theta G_j^I \right) = \exp\left( i\theta Y_{pq} \right).
\end{equation*}
To obtain an analytical form for $U_{j}^I$, first, we compute the square of the operator $Y_{pq}$,
\begin{align}
Y_{pq}^2 
&= -\left(\ket{a_p}\bra{a_q} - \ket{a_q}\bra{a_p}\right)^2 \nonumber \\
&= \left(\ket{a_p}\bra{a_p} + \ket{a_q}\bra{a_q}\right) \nonumber \\
&=P_{pq} \;\; ,
\end{align}
where $P_{pq}$ is the projector onto the subspace spanned by $\{\ket{a_p},\ket{a_q}\}$. Consequently, the higher powers satisfy
\begin{align}
Y_{pq}^{2k} = P_{pq}, \qquad Y_{pq}^{2k+1} = Y_{pq}.
\end{align}
Using the power series expansion of the exponential,
\begin{align}
e^{i\theta Y_{pq}} 
&= \sum_{n=0}^\infty \frac{(i\theta Y_{pq})^n}{n!} \nonumber \\
&= \sum_{k=0}^\infty \frac{(i\theta)^{2k}}{(2k)!} Y_{pq}^{2k}
+ \sum_{k=0}^\infty \frac{(i\theta)^{2k+1}}{(2k+1)!} Y_{pq}^{2k+1} \nonumber \\
&= \left[\sum_{k=0}^\infty \frac{(i\theta)^{2k}}{(2k)!}\right] P_{pq}
+ \left[\sum_{k=0}^\infty \frac{(i\theta)^{2k+1}}{(2k+1)!}\right] Y_{pq} \nonumber \\
&= \cos\theta\, P_{pq} + i\sin\theta\, Y_{pq}.
\end{align}
Choosing $\theta = \frac{\pi}{2}$, we obtain
\begin{equation}
U_{j}^I \left( \theta = \frac{\pi}{2} \right) = i Y_{pq}.
\end{equation}
Now, the interferometric visibility corresponding to this unitary, as defined in Eq.~\eqref{visibility1}, becomes
\begin{align}
V_{j} \left(\theta = \frac{\pi}{2}\right) &= \left|\mathrm{Tr}\left(U_{j}^I \left(\theta = \frac{\pi}{2}\right)\,\rho \right)\right| \nonumber \\
    &= \left|\mathrm{Tr}(i Y_{pq}\,\rho)\right| \nonumber\\
    &= 2 |b_j^I|.
\end{align}

The total visibility is then given by,
\begin{equation}
    V \left(\theta = \frac{\pi}{2}\right) = 2 \sum_j |b_j^I| =\mathcal{M}_{\ell_1}.
\end{equation} 
This clearly indicates that the total visibility $V \left(\theta = \frac{\pi}{2}\right)$ is equal to the $\ell_1$-norm measure of imaginarity, and it vanishes if and only if $b_j^{I}=0$, i.e., the quantum state $\rho$ has zero imaginarity.
This completes the proof.
\end{proof}

For a better illustration of this detection criterion, supporting examples are provided in Appendix \ref{appendixC}.

\begin{table*}
\centering
\begin{tabular}{c c c}
\hline \hline
\textbf{Method} &
\textbf{\makecell{Experimentally\\ measured quantity}} &
\textbf{Detection strategy}
\\
\hline
\makecell{State tomography~\cite{haah2017sample, odonnell2016efficient, d2003quantum_tomography}} &
\makecell{Tomographically complete\\ measurement set} &
\makecell{Evaluate from the\\ reconstructed state}
\\
\makecell{Linear imaginarity witnesses~\cite{liang2026imaginaritywitness, zhang2025imaginarity}} &
\makecell{Expectation value(s) of\\ witness observable(s)} &
Single sufficient witness
\\
\makecell{Bargmann-invariant methods~\cite{fernandes2024unitary}} &
\makecell{State overlaps}&
\makecell{Basis-independent certification \\ of some state in a set of states} 
\\
\textbf{Present work} &
\makecell{Finite set of \\ extended KD moments} &
\makecell{Hierarchy of \\sufficient moment criteria} 
\\
\hline \hline
\end{tabular}
\caption{ 
\justifying
\small 
Comparison between the present detection framework and other existing approaches for detecting quantum imaginarity. The table highlights the experimental quantities that must be measured, the information obtained, and the principal operational features of each method.}
\label{tab:comparison}
\end{table*}

\section{Conclusion}\label{s5}
The role of complex numbers in quantum theory has long attracted conceptual interest, and recent developments have clarified that intrinsically complex features of quantum states can act as a resource enabling advantages in several quantum information processing tasks. Within the resource theoretic framework of quantum imaginarity \cite{hickey2018quantifying,wu2021operational}, such complex features have been linked to operational advantages in tasks including quantum state and channel discrimination \cite{wu2021operational,wu2024resource}, quantum communication \cite{PhysRevA.111.062401}, multiparameter metrology \cite{miyazaki2022imaginarity}, and nonlocal correlations \cite{wei2024nonlocal,datta2025efficient}. These developments make it important to develop reliable and experimentally feasible methods for detecting imaginarity in realistic settings. In this work, we have proposed a framework to detect imaginarity in an unknown quantum state based on experimentally accessible moments of the extended KD distribution.

Existing approaches to imaginarity detection address the problem from different perspectives. Full state tomography provides complete information about the quantum state and allows one to evaluate arbitrary imaginarity measures, but generally requires reconstructing substantially more information than is necessary for mere detection of the resource. Linear-witness approaches can be experimentally economical, although each witness typically detects only a restricted subset of imaginary states. Further, extending the detection capability generally requires additional information about the target state as well as optimization procedures~\cite{liang2026imaginaritywitness,zhang2025imaginarity}. On the other hand, Bargmann-invariant protocols ~\cite{fernandes2024unitary} address the basis-independent question of whether a collection of states admits a simultaneously real representation and therefore probe a conceptually different notion of imaginarity from the standard basis-dependent notion. In contrast, our framework is specifically designed to detect imaginarity of a single quantum state relative to a fixed reference basis. By combining the $Y$-twirl operation with the positivity properties of the extended KD distribution, the problem is reformulated as a moment problem that naturally generates a hierarchy of sufficient criteria. Consequently, the detection capability can be systematically strengthened by considering higher-order conditions within the same framework, without reconstructing the density matrix or constructing new witness observables. The proposed approach therefore provides an alternative state-independent route for detection of imaginarity based on a restricted set of experimentally accessible quantities. A comparison between the different methods is summarized in Table \ref{tab:comparison}.

The framework introduced here also opens several avenues for further investigation. An important direction is the development of moment-based techniques capable of providing experimentally accessible bounds on quantitative measures of imaginarity. Another natural extension is the application of similar ideas to dynamical scenarios, where the generation or evolution of imaginarity under quantum processes could be probed \cite{roden_2016}. More generally, the moment-based methodology explored in this work may offer a practical strategy for detecting other quantum resources through experimentally accessible quantities. Finally, while the proposed moments admit estimation within the classical shadow-tomography framework, providing a rigorous protocol-specific method and copy-complexity analysis remains an important direction for future research.

\section{Acknowledgements}
S.M. and B.M. acknowledge Subrata Bera and Indranil Biswas for initial discussions on imaginarity. B.M. acknowledges the DST INSPIRE fellowship programme for financial support. 
\appendix 
\section{Proof of Theorem \ref{theorem1}} \label{appendixA}
Let $\rho$ be a $d$-dimensional quantum state. We consider its extended Kirkwood–Dirac (KD) distribution $\{Q^{\star}_{ijk}\}$ constructed with respect to two mutually unbiased bases (MUBs). The first two indices are associated with the same basis $\{\ket{a_i}\}_{i=1}^d$, while the third index corresponds to a mutually unbiased basis $\{\ket{b_k}\}_{k=1}^d$. The indices $i,j,k$ each take $d$ possible values, and hence the extended KD distribution consists of $d^3$ elements. We can now arrange these $d^3$ elements along the diagonal of a matrix of dimension $d^3 \times d^3$. Denoting the diagonal entries by $\lambda_1,\lambda_2,\dots,\lambda_{d^3}$, we write
\begin{align}
    D = \mathrm{diag}(\lambda_1,\lambda_2,\dots,\lambda_{d^3}).
\end{align}
Since the extended KD distribution is positive and normalized, its entries satisfy
$\lambda_\ell \ge 0,
\;
\sum_{\ell=1}^{d^3} \lambda_\ell = 1.$
It follows that the diagonal matrix $D$ is positive semidefinite and has unit trace.

We now consider the sequence of extended KD moments 
$\mathbf{r}=(r_1,r_2,\ldots,r_{2m+1})$, defined in Eq.~\eqref{extendedkdmoments}, and construct the corresponding Hankel matrix $H_m(\mathbf{r})$.  
Since these moments arise from a positive distribution supported on $\{\lambda_\ell\}$, the Hankel matrix admits the Vandermonde decomposition
\cite{boley1997vandermonde,tyrtyshnikov1994bad,HeinigRost1984}
\begin{equation}
H_m(\mathbf{r}) = \mathcal{V}_m D \mathcal{V}_m^{T},
\end{equation}
where the Vandermonde matrix $\mathcal{V}_m$ is given by
\begin{align}
\mathcal{V}_m =
\begin{pmatrix}
1 & 1 & \cdots & 1 \\
\lambda_1 & \lambda_2 & \cdots & \lambda_{d^3} \\
\vdots & \vdots & \ddots & \vdots \\
\lambda_1^{m} & \lambda_2^{m} & \cdots & \lambda_{d^3}^{m}
\end{pmatrix}.
\end{align}

To prove the positivity of the Hankel matrix $H_m(\mathbf{r})$, we take an arbitrary real vector 
$\boldsymbol{x}=(x_1,\dots,x_{m+1}) \in \mathbb{R}^{m+1}$. Then
\begin{align}
\boldsymbol{x} H_m(\mathbf{r}) \boldsymbol{x}^T
&= \boldsymbol{x} \mathcal{V}_m D \mathcal{V}_m^T \boldsymbol{x}^T 
= \boldsymbol{z} D \boldsymbol{z}^T,
\end{align}
where $\boldsymbol{z}=\boldsymbol{x}\mathcal{V}_m$. 
Its components are
\begin{align}
z_\ell = \sum_{j=1}^{m+1} x_j \lambda_\ell^{\,j-1},
\qquad
\ell=1,\dots,d^3.
\end{align}
Therefore,
$\boldsymbol{z} D \boldsymbol{z}^T
= \sum_{\ell=1}^{d^3} \lambda_\ell z_\ell^2.$
Since $\lambda_\ell \ge 0$, each term in the sum is non-negative.
Consequently,
\[
\boldsymbol{x} H_m(\mathbf{r}) \boldsymbol{x}^T \ge 0
\]
for all real vectors $\boldsymbol{x}$, which shows that $H_m(\mathbf{r})$ is positive 
semidefinite. Consequently, all principal minors of $H_m(\mathbf{r})$ are non-negative, therefore
\[
\det[H_m(\mathbf{r})] \ge 0.
\]
This completes the proof.

\section{Explicit calculation of extended KD moments and Hankel determinants for the Example in Sec.~\texorpdfstring{Sec.~\ref{sec:example_imaginarity}}{Sec.}}
\label{app:kd_calculations}

For $\theta=\pi/2$, the density operator for the state in Eq.~\eqref{qubit_example} is
\begin{equation}
\rho =
\frac{1}{2} \begin{pmatrix}
1 &  \exp({-i\alpha}) \\
 \exp ({i\alpha}) & 1
\end{pmatrix}.
\end{equation}
Applying the $Y$-twirl map, $ \rho'=\frac12(\rho+\sigma_y\rho\sigma_y)$,
we get,
\begin{equation}
\rho'=
\frac12 \begin{pmatrix}
1 & -i \sin\alpha \\
i \sin\alpha & 1
\end{pmatrix}.
\end{equation}
Next we calculate the extended KD elements as defined in Eq.~\eqref{eq:extended_KD_definition}, for the state $\rho'$, using the computational basis $\{ \ket{a_i}\}$, and the mutually unbiased basis $\{\ket{b_j}\}$ introduced in Eq.~\eqref{mutulally_unbiased_bases},
\begin{align}
&Q^{\star}_{1,1,1}=Q^{\star}_{1,1,2}=Q^{\star}_{2,2,1}=Q^{\star}_{2,2,2}=\frac14, \nonumber \\
&Q^{\star}_{1,2,2}= + \frac{i}{4} e^{ i\beta}\sin\alpha, \quad Q^{\star}_{1,2,1} = - \frac{i}{4} e^{ i\beta}\sin\alpha, \nonumber \\
&Q^{\star}_{2,1,1}= + \frac{i}{4} e^{- i\beta}\sin\alpha, \quad Q^{\star}_{2,1,2} = - \frac{i}{4} e^{ -i\beta}\sin\alpha.
\end{align}
With $X=\sin^2\alpha$, the first seven moments are
\begin{align}
r_1 &= 1, \quad r_2= \frac14 - \frac14 X \cos(2\beta), \quad r_3 = \frac1{16}, \nonumber\\
r_4 &= \frac1{64} + \frac1{64} X^2 \cos(4\beta), \quad r_5 = \frac1{256}, \nonumber\\
r_6 &= \frac1{1024} + \frac1{1024} X^3 \cos(6\beta), \quad r_7 = \frac1{4096}.
\end{align}
We now determine the range of $\beta$ ( with $\alpha \neq 0$) for which the Hankel determinants are negative. The determinant of the first order Hankel matrix is given by,
\begin{equation}
\det(H_1)
=
\frac{X}{2^4}
\left(
2\cos(2\beta)
-
X\cos^2(2\beta)
\right).
\end{equation}
Since $X/2^4>0$, the determinant is negative whenever
\begin{equation}
2\cos(2\beta) - X\cos^2(2\beta) < 0,
\end{equation}
which leads to the condition, 
\begin{align}
\det(H_1) < 0
\quad &\Longleftrightarrow\quad
\cos(2\beta) < 0 \\
\Longrightarrow \beta \in
\left(\frac{\pi}{4},\frac{3\pi}{4}\right)
&\;\,\, \cup\;
\left(\frac{5\pi}{4},\frac{7\pi}{4}\right).
\end{align}
Thus $\det (H_1)$ detects the entire range of imaginarity when $\beta = {\pi}/{2}$, but fails when $\beta=0$ or $\beta =\pi/4$.

Now we consider the determinant of the second order Hankel matrix, which is given by,
\begin{equation}
\det(H_2)
=
-\frac{X^2}{2^{11}}
\left(
4\cos(4\beta)
+
X^2\cos^2(4\beta)
\right).
\end{equation}
Since the prefactor is negative, the determinant is negative whenever
\begin{equation}
4\cos(4\beta)
+
X^2\cos^2(4\beta)
>0.
\end{equation}
This reduces to
\begin{equation}
\det(H_2) < 0
\quad\Longleftrightarrow\quad
\cos(4\beta) > 0.
\end{equation}
That is,
\begin{equation}
\beta \in
\bigcup_{k=0}^{3}
\left(
\frac{k\pi}{2}-\frac{\pi}{8},
\frac{k\pi}{2}+\frac{\pi}{8}
\right),
\end{equation}
which justifies why $\det (H_2)$ successfully detects the entire range of imaginarity when $\beta =0$, but fails for $\beta = \pi/4$.
The determinant of the third-order Hankel matrix is given by,
\begin{equation}
\det(H_3)
=
-\frac{X^4}{2^{18}}
\sin^2(2\beta).
\end{equation}

Since $\sin^2(2\beta)\ge0$, we obtain
\begin{equation}
\det(H_3) < 0
\quad\Longleftrightarrow\quad
\sin^2 (2\beta) > 0.
\end{equation}

Equivalently,
\begin{equation}
\beta \neq 0,\frac{\pi}{2},\pi,\frac{3\pi}{2}.
\end{equation}

This justifies why the determinant of the third-order Hankel matrix successfully detects the imaginarity of $\rho$, when $\beta = \pi/4$.

\section{Shadow estimation of the extended KD moments}\label{app:shadow}

The interferometric scheme discussed in the manuscript provides one experimentally feasible route for realizing the extended KD moments. An alternative realization is provided by the framework of shadow tomography, which aims to estimate linear and nonlinear properties of an unknown quantum state directly from randomized measurements without reconstructing the density operator \cite{aaronson2018shadow}. In contrast to quantum state tomography, where the objective is to reconstruct $\rho$, shadow tomography predicts only the desired functions of $\rho$, thereby significantly reducing the experimental resources whenever only a limited number of quantities are of interest~\cite{huang2020predicting}.

A key observation is that the extended KD moments introduced in the present work belong exactly to the class of nonlinear functions considered in Ref.~\cite{huang2020predicting}. The $n$th-order moment admits the multicopy representation
\begin{equation*}
r_n'
=
\mathrm{Tr}\!\left[S_n(\rho')^{\otimes n}\right],
\end{equation*}
where the operator
\(
S_n=\sum_{i,j,k}\left(\Pi_k^{b}\Pi_j^{a}\Pi_i^{a}\right)^{\otimes n},
\)
is completely determined by the chosen measurement bases (see Eq.~\eqref{nonlinear_KD}). Consequently, each KD moment is a nonlinear function of the quantum state that can be expressed as the expectation value of a fixed operator acting on a finite number of copies of the state. Thus, the estimation of KD moments is not a separate estimation problem, but rather a direct instance of the general nonlinear-observable estimation framework developed in the shadow-tomography literature.

In a standard classical-shadow protocol \cite{huang2020predicting}, one starts with $N$ identical copies of the unknown state $\rho'$. During each experimental run, a unitary $U$ is sampled independently from a predetermined ensemble of unitaries, after which the rotated state $U\rho' U^\dagger$ is measured in the computational basis. If the measurement outcome is the basis state $\ket{b}$, a classical snapshot is constructed using the inverse channel as
\begin{equation}
\hat{\rho}
=
\mathcal{M}^{-1}
\!\left(
U^\dagger |b\rangle\langle b| U
\right).
\end{equation}
Here $\mathcal{M}$ denotes the measurement channel associated with the chosen ensemble. Repeating this procedure independently produces the classical shadow
\begin{equation}
\mathcal{S}(\rho';N)
=
\left\{
\hat{\rho}_1,\hat{\rho}_2,\ldots,\hat{\rho}_N
\right\},
\end{equation}
which is a collection of classical snapshots. The expectation over the snapshots satisfy the property,
$\mathbb{E}[\hat{\rho}]
=
\rho'.$
For nonlinear functions of the form given in Eq.~\eqref{nonlinear_KD}, unbiased estimators are obtained by combining $n$ independent snapshots through suitable multicopy estimators \cite{PhysRevLett.125.200501,cieslinski2024analysing, liu_replica_shadow}.

The resource requirements of this procedure are naturally described within the standard shadow-tomography framework. For a fixed collection of target moments $\{r_2',r_3',\ldots,r_m'\}$, one considers their estimation up to a prescribed additive error and confidence level. Existing shadow-tomography results show that the required number of copies is governed by the variance (or equivalently, the shadow norm) of the corresponding estimator together with the number of target observables being estimated simultaneously \cite{aaronson2018shadow,huang2020predicting,PhysRevLett.125.200501,cieslinski2024analysing}. Since the order of the moments entering the Hankel hierarchy is fixed independently of the Hilbert-space dimension, and only a finite number of moments is required to evaluate any chosen member of the hierarchy, the proposed framework is directly compatible within these existing shadow-estimation protocols. We emphasize that the contribution of the present work is to identify the extended KD moments associated with quantum imaginarity as experimentally accessible nonlinear functions. The copy complexity for a particular shadow protocol follows from the corresponding shadow-estimation framework and is beyond the scope of this present work.

\section{Illustrative examples in support of Theorem \ref{theorem3}} \label{appendixC}
Here, we illustrate the detection criterion presented in Theorem~\ref{theorem3} considering qubit and qutrit states and show how the interferometric visibility directly reveals the presence of imaginary coherence. 

\subsection*{Qubit case (\texorpdfstring{$d=2$}{d=2})}

In this case there is only one antisymmetric operator,
\begin{equation}
G^I = i(\ket{0}\bra{1}-\ket{1}\bra{0}).
\end{equation}

Consider a general qubit density matrix
\begin{equation}
\rho =
\begin{pmatrix}
a & b\\
b^* & 1-a
\end{pmatrix},
\end{equation}
where $a\in[0,1]$ and $|b|^2\leq a(1-a)$. Writing $b=x+iy$ with $x,y\in\mathbb{R}$, the expectation value of $G^I$ becomes
\begin{equation}
\Tr(G^I \rho) = i(b^*-b) = 2y .
\end{equation}

If the unitary in the interferometer is chosen as $i G^I$, the corresponding visibility becomes, $V=2|y|$. Thus the visibility is nonzero if and only if the imaginary part of the coherence element $\rho_{01}$ is nonzero. The $\ell_1$-norm measure of imaginarity for this state is
\begin{equation}
\mathcal{M}_{{\ell}_1}(\rho)=2|\mathrm{Im}(b)|=2|y|,
\end{equation}
which coincides with the visibility obtained in the interferometric scheme. Hence the interferometric visibility directly quantifies the imaginarity of the qubit state.

\subsection*{Qutrit case (\texorpdfstring{$d=3$}{d=3})}

For $d=3$, the antisymmetric operators are
\begin{align}
G^I_1 &= i(\ket{0}\bra{1}-\ket{1}\bra{0}),\\
G^I_2 &= i(\ket{0}\bra{2}-\ket{2}\bra{0}),\\
G^I_3 &= i(\ket{1}\bra{2}-\ket{2}\bra{1}).
\end{align}

Consider a general qutrit density matrix
\begin{equation}
\rho =
\begin{pmatrix}
a & b & c\\
b^* & d & e\\
c^* & e^* & f
\end{pmatrix},
\end{equation}
where the matrix elements satisfy the conditions as in Eq.~\eqref{positivity_rho}.
Writing
\(
b=x_1+i y_1, \quad
c=x_2+i y_2, \quad
e=x_3+i y_3,
\)
the expectation values of the corresponding unitary operators become
\begin{align}
\Tr(iG^I_j \rho) = 2iy_j, \quad \; j=1,2,3
\end{align}

Accordingly, the visibilities obtained from the interferometric scheme are $V_j = 2|y_j|$,
using Theorem \ref{theorem3}, the total visibility becomes
\begin{equation}
V = V_{1}+V_{2}+V_{3}
= 2\left(|y_1|+|y_2|+|y_3|\right).
\end{equation}

Moreover, the $\ell_1$-norm measure of imaginarity for a qutrit state is
\begin{equation}
\mathcal{M}_{\ell_1}(\rho) = 2\left(|y_1|+|y_2|+|y_3|\right),
\end{equation}
which matches the total visibility obtained from the interferometric scheme. Therefore, the presence of at least one nonzero imaginary coherence element leads to a nonzero visibility, which is consistent with Theorem~\ref{theorem3}.

\bibliography{main}
\end{document}